\documentclass{article}
\DeclareUnicodeCharacter{2212}{\ensuremath{-}}
\usepackage{amsmath}
\usepackage{amssymb}
\usepackage[pdftex,colorlinks]{hyperref}
\usepackage{multirow}
\usepackage{enumerate}
\usepackage{cases}
\usepackage{amsthm}
\usepackage{float}
\usepackage{booktabs}
\usepackage{makecell}
\usepackage{graphicx}
\usepackage{mathrsfs}
\usepackage{amsmath,amsfonts,amssymb,graphicx,mathtools,bm,tcolorbox,relsize}
\usepackage{color}
\usepackage{algorithm} 
\usepackage{algorithmicx}
\usepackage{algpseudocode} 
\usepackage[misc]{ifsym}
\usepackage{caption}
\usepackage{subfigure}
\usepackage{array}
\usepackage{authblk} 
\usepackage{cite}
\newcommand\keywords[1]{\textbf{Keywords}: #1}

\newtheorem{definition}{Definition}
\newtheorem{remark}{Remark}
\newtheorem{lemma}{Lemma}
\newtheorem{proposition}{Proposition}
\newtheorem{theorem}{Theorem}

\title{\textbf{The exact lower bound of CNOT-complexity for fault-tolerant quantum Fourier transform}}

\author[1,2,3]{Qiqing Xia}
\author[4]{Huiqin Xie}
\author[1,3]{Li Yang  \thanks{Corresponding author: yangli@iie.ac.cn}}
\affil[1]{Institute of Information Engineering, Chinese Academy of Sciences, Beijing, China}
\affil[2]{School of Cyber Security, University of Chinese Academy of Sciences, Beijing, China}
\affil[3]{Key Laboratory of Cyberspace Security Defense, Beijing, China}
\affil[4]{ Beijing Electronic Science and Technology Institute, Beijing, China}
\date{}

\begin{document}

\maketitle

\smallskip
\hrule
\smallskip

\begin{abstract}
The quantum Fourier transform (QFT) is a crucial subroutine in many quantum algorithms. In this paper, we study the exact lower bound problem of CNOT gate complexity for fault-tolerant QFT. First, we consider approximating the ancilla-free controlled-$R_k$ in the QFT logical circuit with a standard set of universal gates, aiming to minimize the number of T gates. Various single-qubit gates are generated in addition to CNOT gates when the controlled-$R_k$ is decomposed in different ways, we propose an algorithm that combines numerical and analytical methods to exactly compute the minimum T gate count for approximating any single-qubit gate with any given accuracy. Afterwards, we prove that the exact lower bound problem of T gate complexity for the QFT is NP-complete. Furthermore, we provide the transversal implementation of universal quantum gates and prove that it has the minimum number of CNOT gates and analyze the minimum CNOT count for transversally implementing the T gate. We then exactly compute the exact lower bound of CNOT gate complexity for fault-tolerant QFT with the current maximum fault-tolerant accuracy $10^{−2}$. Our work can provide a reference for designing algorithms based on active defense in a quantum setting.
\end{abstract}


\bigskip
\hrule
\smallskip

\section{Introduction}
A lower bound exists on the time for quantum gates due to physical limitations, making it necessary to study the optimal implementation of quantum circuits. Especially, the quantum Fourier transform (QFT) \cite{jozsa1998quantum,nielsen2010quantum} is a key subroutine in many quantum algorithms, such as Shor's algorithm for solving integer factorization and discrete logarithm problems \cite{shor1994algorithms,shor1997polynomial}, quantum amplitude estimation \cite{brassard2002quantum}, phase estimation \cite{kitaev1995quantum}, solving systems of linear equations \cite{harrow2009quantum}, quantum arithmetic \cite{ruiz2017quantum} and various fast quantum algorithms for solving hidden subgroup problems \cite{simon1997power,ettinger2000quantum,grigni2001quantum,hallgren2000normal,mosca1998hidden}. The efficient execution of the QFT on quantum computers is crucial for the success of algorithms. One of its significant applications is to break the cryptosystems such as RSA \cite{rivest1978method} and ElGamal \cite{elgamal1985public}. One of the most feasible candidates for these large-scale quantum computations is fault-tolerant quantum computation (FTQC) \cite{shor1996fault,preskill1998fault,gottesman1998theory,gottesman2010introduction,steane1999efficient}. FTQC must further rely on the Clifford gates and T gates, and any single-qubit gate can be approximated to an arbitrary accuracy using a circuit composed of these gates. However, the T gate requires ancillary states and CNOT gates for transversal implementation, which is relatively expensive \cite{nielsen2010quantum,fowler2009high}. The number of T gates can serve as a first-order approximation of the resources for physically implementing a quantum circuit.

Currently, there are two main methods for optimizing the number of T gates in QFT circuits. One method is to define an approximate version of the QFT (AQFT) \cite{coppersmith2002approximate,barenco1996approximate} by removing all rotation gates with angles below a certain threshold, thereby optimizing the number of T gates for the AQFT \cite{nam2020approximate, park2022t}. The other is to decompose the controlled-$R_k$ gates in the QFT circuit and use gate synthesis or state distillation methods to approximate the single-qubit gates. In \cite{goto2014resource}, this work compares the resources for gate synthesis and state distillation, concluding that gate synthesis is a better method for implementing the fault-tolerant QFT. The optimal method for gate synthesis is provided by Solovay-Kitaev algorithm \cite{kitaev2002classical}, and the number of T gates for approximating any single-qubit gate is $\mathcal{O}(log^{3.97}/\epsilon)$. Based on this, \cite{fowler2011constructing} uses a numerical method to find the fault-tolerant approximation properties of the phase rotation gates in exponential time. The optimization methods for the number of T gates required to approximate any single-qubit gate include probabilistic algorithms with feedback \cite{bocharov2015probabilistic}, efficient random algorithms \cite{selinger2012efficient}, using Repeat-Until-Success circuits \cite{bocharov2015efficient,paetznick2013repeat}, using auxiliary qubits \cite{kliuchnikov2013asymptotically,kim2018efficient}, efficient approximation algorithms achieving accuracy $10^{-15}$ \cite{kliuchnikov2015practical}, rounding off a unitary to a unitary over the ring $\mathcal{Z}[i,1/\sqrt{2}]$ \cite{kliuchnikov2013synthesis}, and probabilistic approximation for Z-rotation gates based on the gridsynth algorithm \cite{ross2016optimal}.

These related works focus on optimizing the number of logical T gates associated with the QFT using approximation algorithms. However, the specific implementation of a quantum computer is influenced by physical factors, as it is fundamentally a physical system. In FTQC, the CNOT gate is regarded as the main cause of quantum errors \cite{xia2024analysis,yang2020effect} and has a longer execution time than single-qubit gates. The computational complexity of quantum algorithms can be reduced to the number of CNOT gates \cite{knill2005quantum}, and minimizing the number of CNOT gates is typically regarded as a secondary optimization objective, such as optimizing the number of CNOT gates in the quantum circuit for Shor's algorithm \cite{liu2021cnot,liu2023minimizing}. However, the problem of CNOT-optimal quantum circuit synthesis over gate sets consisting of CNOT and phase gates is NP-complete \cite{amy2019controlled}. In some cases, optimizations aimed at reducing the number of T gates can result in a significant increase in the number of CNOT gates. Fortunately, we have found a method to make up for the explosion of CNOT gates with cost savings by optimizing T gates. Proving the ancilla-assisted gate synthesis with the minimum number of CNOT gates is challenging, while it is provable for the ancilla-free gate synthesis. Thus, in this way, we can obtain the exact lower bound of CNOT gate complexity for fault-tolerant QFT, which is different from previous works. To our knowledge, there has not been such an analysis before. This analysis is of great significance and can provide a reference for active defense in a quantum setting.

\textbf{Our contributions} 
In this paper, we consider the logical resources and fault-tolerance resources of the QFT. The logical resources are measured by the number of T gates, which is called the T-count. Since the fault-tolerant implementation of the T gate relies on CNOT gates, the fault-tolerance resources are measured by the number of CNOT gates, which is called the CNOT-count. In more detail, our main contributions are as follows:
\begin{itemize}
    \item  First, we find an ancilla-free gate synthesis method for a controlled-$R_k$ gate with the minimum T-count and we have provided the corresponding proof. Interestingly, this method also achieves the minimum fault-tolerant CNOT-count. When the controlled-$R_k$ is decomposed in different ways, it generates different single-qubit gates in addition to CNOT gates. Since current quantum computers cannot implement all single-qubit gates, as a generalization we propose an algorithm to exactly compute the minimum T-count for approximating any single-qubit gate with any given accuracy by using Hadamard gates and T gates. This algorithm uses analytical methods to avoid traversing all quantum states in the normalized space and uses numerical methods to determine whether the error conditions are satisfied. Afterwards, we prove that the exact lower bound problem of the T-count for the QFT is at least as hard as the K-SAT problem and it is an NP-complete problem.
    
    \item Furthermore, we provide a transversal implementation of Z-rotation gates satisfying certain conditions with the minimum CNOT-count and analyze the minimum CNOT-count for transversally implementing the T gate. We then compute the exact lower bound of CNOT gate complexity for fault-tolerant QFT with different input lengths at the current maximum fault-tolerant accuracy $10^{-2}$ \cite{fowler2009high, knill2005quantum}. In particular, we estimate the lower bound of the effective execution time of the QFT based on Steane code on ion trap computers, which can provide a reference for quantum computation based on the QFT. Finally, we discuss that determining the best possible value of $c$ in $\mathcal{O}(\log^c(1/\epsilon))$ implied by the Solovay–Kitaev theorem is at least NP-hard, and the circuit optimization problem is at least QMA-hard.
\end{itemize}

\section{Preliminaries}
\subsection{Common quantum gates}
We briefly recall Clifford and T quantum gates, including universal gates and Pauli-X (Y, Z) gates. Their circuit symbols and matrix representations are shown in Figure \ref{fig:fuqg}.
\begin{figure}[H]
    \centering
    \includegraphics[width=0.75\textwidth]{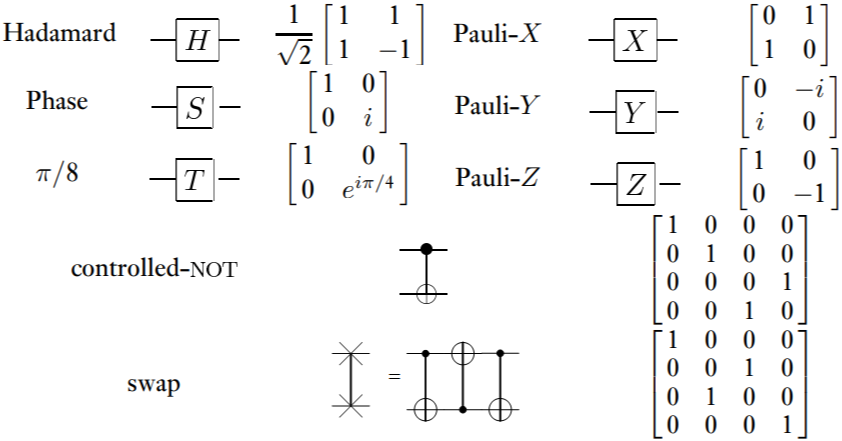}
    \caption{Clifford and T quantum gates}
    \label{fig:fuqg}
\end{figure}
Here, \{H, S, T, CNOT\} is called the standard set of universal gates. When the Pauli matrices are exponentiated, they will produce three types of useful unitary matrices, namely, the rotation operators around the $\hat{x}$, $\hat{y}$ and $\hat{z}$ axes, defined as follows:

\begin{equation}\label{R_xyz-eq}
\begin{aligned}
&R_x(\theta)\equiv e^{-i\theta X/2}=cos \frac{\theta}{2}I-isin \frac{\theta}{2}X=\begin{bmatrix}
cos \frac{\theta}{2} & -isin \frac{\theta}{2} \\
-isin \frac{\theta}{2} & cos \frac{\theta}{2}
\end{bmatrix},   \\
&R_y(\theta)\equiv e^{-i\theta Y/2}=cos \frac{\theta}{2}I-isin \frac{\theta}{2}Y=\begin{bmatrix}
cos \frac{\theta}{2} & -sin \frac{\theta}{2} \\
sin \frac{\theta}{2} & cos \frac{\theta}{2}
\end{bmatrix},   \\
&R_z(\theta)\equiv e^{-i\theta Z/2}=cos \frac{\theta}{2}I-isin \frac{\theta}{2}Z=\begin{bmatrix}
e^{\frac{-i\theta}{2}} & 0 \\
0 & e^{\frac{i\theta}{2}}
\end{bmatrix},
\end{aligned}
\end{equation}
if rotating by an angle $\theta$ around the axis $\hat{n}=(n_x,n_y,n_z)$, then the rotation operator can be denoted as
    \begin{equation}\label{R_n-eq}
        R_{\hat{n}}(\theta)=\cos{\frac{\theta}{2}}I-i\sin{\frac{\theta}{2}}(n_xX+n_yY+n_zZ),
    \end{equation}
these rotation operators play a crucial role in decomposing controlled unitary operators.

\subsection{Approximating controlled single-qubit unitary operators}
We first introduce a lemma on the decomposition of a single-qubit unitary operator and a theorem on the decomposition of its controlled operator, as proven in \cite{nielsen2010quantum}:
\begin{lemma}\label{Unitary operator decomposition-lemma}
Suppose $U$ is a unitary operator on a single qubit, then there exist real numbers $\alpha,\beta,\gamma$ and $\delta \in [0,2\pi)$, such that
\begin{equation}\label{Unitary operator decomposition-lemma-eq1}
    U=e^{i\alpha}R_z(\beta)R_y(\gamma)R_z(\delta) =\begin{bmatrix}
    e^{i(\alpha-\beta/2-\delta/2)}\cos{\frac{\gamma}{2}} & -e^{i(\alpha-\beta/2+\delta/2)}\sin{\frac{\gamma}{2}}\\
    e^{i(\alpha+\beta/2-\delta/2)}\sin{\frac{\gamma}{2}} & e^{i(\alpha+\beta/2+\delta/2)}\cos{\frac{\gamma}{2}}
    \end{bmatrix}.
\end{equation}

It can be extended to a more general case: suppose $\hat{m}$ and $\hat{n}$ are non-parallel real unit vectors in the three-dimensional space, then $U$ can be written as
\begin{equation}\label{Unitary operator decomposition-lemma-eq2}
U=e^{i\alpha}R_{\hat{n}}(\beta)R_{\hat{m}}(\gamma)R_{\hat{n}}(\delta), 
\end{equation}
for appropriate choices of $\alpha,\beta,\gamma$ and $\delta$.
\end{lemma}

For any controlled single-qubit unitary operator controlled-$U$, up to a global phase, it can be decomposed into two CNOT gates and three single-qubit unitary operators, the theorem is as follows:
\begin{theorem}\label{Controlled-$U$ decomposition-theorem}
Suppose $U$ is a unitary gate on a single qubit, then there exist single-qubit unitary operators $P\equiv e^{i\alpha/2}R_z(\alpha), A\equiv R_z(\beta)R_y(\gamma/2),B\equiv R_y(-\gamma/2)R_z(-(\delta+\beta)/2),C\equiv R_z((\delta-\beta)/2)$, such that $ABC=I$ and $U=e^{i\alpha}AXBXC$, where $\alpha$ is a global phase factor. Then, the controlled-$U$ operation is $C(U)=(P\otimes A) \cdot CNOT \cdot (I\otimes B) \cdot CNOT \cdot (I\otimes C)$. The circuit implementation is shown in Figure \ref{fig:citco}.
\begin{figure}[H]
\centerline{\includegraphics[width=0.75\textwidth]{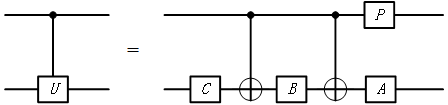}}
  \caption{Circuit implementing the controlled-$U$ operation}
   \label{fig:citco}
\end{figure}
\end{theorem}

The set of universal gates \{H, S, T, CNOT\} is discrete, while the set of unitary operations is continuous. Approximating any unitary operation using this discrete set will inevitably introduce errors, as defined in \cite{nielsen2010quantum}:
\begin{definition}\label{Error-Unitary operator}
    Let $U$ and $V$ be two unitary operators on the same state space, where $U$ is the desired target unitary operator and $V$ is the unitary operator actually implemented. Define the error when $V$ is implemented instead of $U$ as
    \begin{equation}\label{Error-Unitary operator-eq}
        E(U,V)=\max_{|\psi\rangle}||(U-V)|\psi\rangle||
    \end{equation}
where the maximum takes all normalized quantum states $|\psi\rangle$ in the state space.
\end{definition}

The Solovay–Kitaev theorem \cite{kitaev2002classical} is one of the most important fundamental results in the theory of quantum computation. It shows that for any single-qubit gate $U$ and given any $\epsilon>0$, $U$ can be approximated to an accuracy $\epsilon$ with $\mathcal{O}(log^c(1/\epsilon))$ finite gates.
\begin{theorem}\label{sk-theorem}
    Let $SU(2)$ be the set of all single-qubit unitary matrices with determinant 1, and $\mathcal{G}$ be a finite set of elements in $SU(2)$ including its own inverses, used to simulate other single-qubit gates. Let $g_1\cdots g_l$ $(g_i\in \mathcal{G}, i=1,\cdots,l)$ be a word of length $l$ from $\mathcal{G}$, $\mathcal{G}_l$ be the set of all words of length at most $l$, and $\langle \mathcal{G} \rangle$ be the set of all words with finite length. If $\langle \mathcal{G} \rangle$ is dense in $SU(2)$, then $\mathcal{G}_l$ is an $\epsilon$-net in $SU(2)$ for $l=\mathcal{O}(log^c(1/\epsilon))$, where $c \approx 2$ when using the measure $E(\cdot,\cdot)$, or $c \approx 4$ when using the trace distance $D(\cdot,\cdot)$, as $D(\cdot,\cdot)=2E(\cdot,\cdot)$. In our work, we use the measure $E(\cdot,\cdot)$ as show in Definition \ref{Error-Unitary operator}.
\end{theorem}

\subsection{The quantum Fourier transform}
\begin{definition}\label{QFT}
The QFT is a linear operator on an orthonormal basis $|0\rangle,\cdots,|N-1\rangle$, and its action on the basis states is
\begin{equation}\label{QFT-eq1}
    |j\rangle \xrightarrow{QFT(N)} \frac{1}{\sqrt{N}}\sum_{k=0}^{N-1}e^{2\pi ijk/N}|k\rangle.
\end{equation}
Equivalently, the action on an arbitrary state can be written
\begin{equation}\label{QFT-eq2}
    \sum_{j=0}^{N-1}x_j|j\rangle \xrightarrow{QFT(N)} \frac{1}{\sqrt{N}}\sum_{k=0}^{N-1}y_k|k\rangle,
\end{equation}
where $y_k$ is the discrete Fourier transform of the amplitudes $x_j$, $y_k=\sum_{j=0}^{N-1}x_je^{2\pi ijk/N}$.
\end{definition}

The QFT is a unitary transformation, and thus it can be implemented as a dynamic process on a quantum computer as follows:
\begin{center}
\fbox{\begin{minipage}{1\textwidth}
\small
\textbf{Useful product representation of the QFT}
\begin{enumerate}
    \item  First, suppose $N=2^n$, and represent $j$ in binary as $j_1j_2\cdots j_n$, then $j=j_12^{n-1}+j_22^{n-2}+\cdots+j_n$.
    
    \item After performing the QFT on the basis state $|j\rangle$, then we can obtain the useful product representation:
    \begin{equation}\label{Useful product form-eq}
    \begin{aligned}
        |j_1,\cdots,j_n\rangle \xrightarrow{QFT(N=2^n)}\frac{1}{2^{n/2}}\left(|0\rangle+e^{2\pi i0.j_n}|1\rangle \right)\left(|0\rangle+e^{2\pi i0.j_{n-1}j_n}|1\rangle \right)\cdots \left(|0\rangle+e^{2\pi i0.j_1j_2\cdots j_n}|1\rangle \right),
    \end{aligned}
    \end{equation}
    where $0.j_{\ell}j_{\ell+1}\cdots j_m=\frac{j_{\ell}}{2}+\frac{j_{\ell+1}}{2^2}+\cdots+\frac{j_m}{2^{m-\ell+1}}$ is a binary fraction.

    \item Each tensor product component can be implemented by a controlled single-qubit unitary operation (controlled-$R_k$). Here, the single-qubit gate $R_k$ denotes the unitary transformation
     \begin{equation}\label{controlled-Rk-eq}
         R_k=\begin{bmatrix} 
         1&0 \\ 
         0&e^{2\pi i/2^k}       
         \end{bmatrix}.
     \end{equation}
\end{enumerate}
\end{minipage}}
\end{center}

According to the useful product representation of the QFT in Eq. (\ref{Useful product form-eq}), we introduce the construction of effective quantum circuits for the QFT \cite{nielsen2010quantum}, as shown in Figure \ref{fig:eqcftqft}.
\begin{figure}[H]
\centerline{\includegraphics[width=0.98\textwidth]{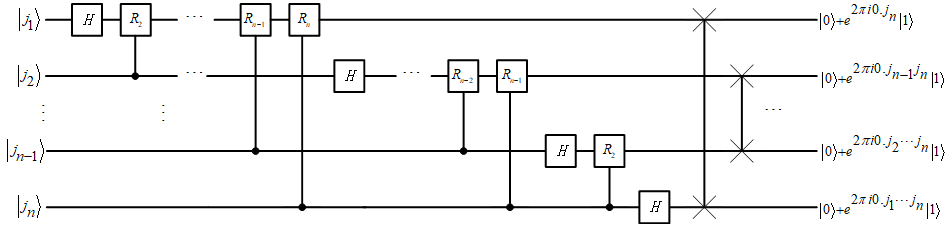}}
  \caption{Efficient quantum circuit for the QFT}
   \label{fig:eqcftqft}
\end{figure}

As can be seen from Figure \ref{fig:eqcftqft}, the efficient quantum circuit for the QFT requires $n$ Hadamard gates, $n(n-1)/2$ controlled-$R_k$ gates, and $\lfloor \frac{n}{2} \rfloor$ swap gates.

\section{The exact lower bound problem of T-count for the QFT}
Although the standard QFT circuit is typically implemented in a specific way, as shown in Figure \ref{fig:eqcftqft}, the QFT circuit is not unique and can achieve the same function through equivalent circuit transformations. Therefore, all QFT circuits must be equivalent to the circuit shown in Figure \ref{fig:eqcftqft}, and thus we can aim to analyze the exact lower bound of the T-count for the QFT based on the circuit in Figure \ref{fig:eqcftqft}.
\subsection{Algorithm for approximating the ancilla-free controlled-$R_k$ with the minimum T-count}
Since current quantum computers cannot implement all single-qubit gates, the ancilla-free gate synthesis approximates single-qubit gates with the Hadamard gate and T gate. To approximate the controlled-$R_k$ operation, the ancilla-free controlled-$R_k$ can be decomposed into single-qubit gates in addition to CNOT gates. However, when different decomposition methods are used, the single-qubit gates decomposed from the controlled-$R_k$ are also different. To approximate different single-qubit gates, we introduce a theorem from \cite{nielsen2010quantum}, for which we have revised the relevant content. The theorem states that any single-qubit unitary operation can be approximated to arbitrary accuracy using the Hadamard gate and T gate, as follows:
\begin{theorem}\label{HTHT-theorem}
    Up to an unimportant global phase, $T$ has the same effect as $R_z(\pi/4)$, and $HTH$ has the same effect as $R_x(\pi/4)$. Combining these two operations, according to Eq. (\ref{R_n-eq}), we obtain
    
    \begin{equation}\label{HTHT-theorem-eq1}
    \begin{aligned}
        T(HTH)&=e^{\frac{i\pi}{4}}R_z(\pi/4)R_x(\pi/4)\\
        &=e^{\frac{i\pi}{4}}\left(\cos{\frac{\pi}{8}}I-i\sin{\frac{\pi}{8}}Z\right)\left(\cos{\frac{\pi}{8}}I-i\sin{\frac{\pi}{8}}X\right)\\
        &=e^{\frac{i\pi}{4}}\left((\cos{\frac{\pi}{8}})^2I-i\sin{\frac{\pi}{8}}\left(\cos{\frac{\pi}{8}}(X+Z)+\sin{\frac{\pi}{8}}Y\right)\right)\\
        &=e^{\frac{i\pi}{4}}\begin{bmatrix} \frac{1+e^{-i\pi/4}}{2}&\frac{e^{-i\pi/4}-1}{2}\\ \frac{1-e^{i\pi/4}}{2}&\frac{1+e^{i\pi/4}}{2} \end{bmatrix}.
    \end{aligned}
    \end{equation}

Using only the Hadamard gate and T gate, the rotation operator $R_{\hat{n}}(\theta)=R_z(\pi/4)R_x(\pi/4)$ can be constructed to approximate any rotation operator around the $\hat{n}$ axe, where $\cos{\frac{\theta}{2}}=(\cos{\frac{\pi}{8}})^2$, $\hat{n}=\frac{\sin{\frac{\pi}{8}}}{\sin{\frac{\theta}{2}}}\left(\cos{\frac{\pi}{8}},\sin{\frac{\pi}{8}},\cos{\frac{\pi}{8}}\right)$. Let $R_{\hat{m}}(\theta))=HR_{\hat{n}}(\theta)H$, $\hat{m}=\frac{\sin{\frac{\pi}{8}}}{\sin{\frac{\theta}{2}}}\left(\cos{\frac{\pi}{8}},-\sin{\frac{\pi}{8}},\cos{\frac{\pi}{8}}\right)$, $\hat{n}$ and $\hat{m}$ are real unit vectors that are not parallel in the three-dimensional space. According to Eq. (\ref{Unitary operator decomposition-lemma-eq2}), there exist appropriate positive integers $n_1,n_2,n_3$,
\begin{equation}\label{HTHT-theorem-eq2}
    E(U,R_{\hat{n}}(\theta))^{n_1}HR_{\hat{n}}(\theta)^{n_2}HR_{\hat{n}}(\theta)^{n_3})<\epsilon,
\end{equation}
where $\epsilon$ is the desired accuracy. For any given single-qubit operator $U$ and any $\epsilon>0$, a quantum circuit consisting only of Hadamard gates and T gates can approximate $U$ within $\epsilon$.
\end{theorem}

Based on Theorem \ref{HTHT-theorem}, we propose an algorithm for exactly computing the minimum resources required to approximate any single-qubit unitary operator $U$ with a given accuracy using $V=R_{\hat{n}}(\theta)^{n_1}HR_{\hat{n}}(\theta)^{n_2}H$ $R_{\hat{n}}(\theta)^{n_3}$. The algorithm is as follows:
\begin{algorithm}[H]
\caption{Compute the minimum resources for single-qubit unitary operators $U$}
\label{algorithm1}
\begin{algorithmic}[1]
\State \textbf{Input:} Rotation operators $R_{\hat{n}}(\theta),U$, integer $N_0$, accuracy $\epsilon$.
\State \textbf{Output:} Minimum resources $\min n_{sum}$ for single-qubit unitary operators $U$.
    \State Let $\min n_{sum}\xleftarrow{} \infty$;
    \For{$n_{sum}$ from 3 to $N_0$}
        \For{$n_1$ from 1 to $n_{sum}-2$}
            \For{$n_2$ from 1 to $n_{sum}-n_1-1$}
                \State compute $n_3 = n_{sum} - n_1 - n_2$;
                \State $V=R_{\hat{n}}(\theta)^{n_1}HR_{\hat{n}}(\theta)^{n_2}HR_{\hat{n}}(\theta)^{n_3}$;
                \State $E(U,V)=\max\limits_{|\psi\rangle}||(U-V)|\psi\rangle||$, for $|\psi\rangle$ in the Bloch sphere;
                \If{$E(U,V)< \epsilon$}
                    \State $\min n_{sum}=n_{sum}$;
                    \State out of all the loop;
                \EndIf
            \EndFor
        \EndFor
    \EndFor
    \State \Return $\min n_{sum}$.
\end{algorithmic}
\end{algorithm}

\begin{remark}\label{rm1}
In Algorithm \ref{algorithm1}, the minimum resources are measured by the minimum of $n_1+n_2+n_3$. The minimum T-count is twice $\min n_{sum}$, Since $R_{\hat{n}}(\theta)=THTH,R_{\hat{m}}(\theta)=HTHT$.
\end{remark}

Since a computer is a discrete computing device, it is impossible to traverse all $|\psi\rangle$ in the normalized space. To avoid this problem and ensure that our algorithm is exact, we combine numerical and analytical methods to perform a mathematical transformation in step 9 of Algorithm \ref{algorithm1}, transforming this traversal problem into an extremum determination problem. The analytical method can avoid traversing all quantum states in the normalized space, and the numerical method can be used to determine whether the error conditions are satisfied. When the accuracy of the extremum computed by computers is much smaller than $\epsilon$, step 10 will not result in any misjudgment. The transformation is as follows:
\begin{center}
\fbox{\begin{minipage}{1\textwidth}
\small
\textbf{Mathematical computation in step 9 of Algorithm \ref{algorithm1}}
\begin{enumerate}
    \item  Let $|\psi\rangle=\cos{\frac{\omega}{2}}|0\rangle+e^{i\varphi}\sin{\frac{\omega}{2}}|1\rangle, \omega\in [0,\pi],\varphi\in [0,2\pi)$, then $|\psi\rangle=\begin{bmatrix}
    \cos{\frac{\omega}{2}}\\
    e^{i\varphi}\sin{\frac{\omega}{2}}
    \end{bmatrix}$.
    
    \item Let $f(\omega, \varphi)=||(U-V)|\psi\rangle||$, where $f(\omega, \varphi)$ is a continuous function over the domain. Then, compute the norm of the two-dimensional vector $(U-V)|\psi\rangle$ for each set of $n_1,n_2,n_3$.

    \item Compute the partial derivatives $\frac{\partial f}{\partial \omega}$ and $\frac{\partial f}{\partial \varphi}$ of the function $f(\omega, \varphi)$ with respect to $\omega$ and $\varphi$ such that $\frac{\partial f}{\partial \omega}=0$ and $\frac{\partial f}{\partial \varphi}=0$.
    \begin{itemize}
        \item If $\frac{\partial f}{\partial \omega}=0$ has no solution, then $\frac{\partial f}{\partial \omega}>0$ (or$<0$) always holds, and the local maximum point is the boundary point $(\pi,\varphi_0)$ (or $(0,\varphi_0)$).
        \item If $\frac{\partial f}{\partial \varphi}=0$ has no solution, then $\frac{\partial f}{\partial \varphi}>0$ (or$<0$) always holds, and $f(\omega, 0)\neq f(\omega, 2\pi)$, which contradicts $f(\omega, 0)=f(\omega, 2\pi)$! Therefore, $\frac{\partial f}{\partial \varphi}=0$ must have a solution.
        \item If both $\frac{\partial f}{\partial \omega}=0$ and $\frac{\partial f}{\partial \varphi}=0$ have solutions, then the possible extremum points $(\omega_0,\varphi_0)$ are found.
    \end{itemize}
    
    \item Compute the second-order partial derivatives of the function $f$ at $(\omega_0,\varphi_0)$ and construct the Hessian matrix $D$ as follows:
    \begin{equation}\label{algorithm-eq}
        D=\begin{bmatrix}
        \frac{\partial^2 f}{\partial \omega^2} & \frac{\partial^2 f}{\partial \omega \partial \varphi}\\
        \frac{\partial^2 f}{\partial \varphi \partial \omega} & \frac{\partial^2 f}{\partial \varphi^2}
        \end{bmatrix}.
    \end{equation}
    
    \item Compute the determinant $\det(D)$ of the Hessian matrix $D$ and determine the extremum point $(\omega_0,\varphi_0)$:
    \begin{itemize}
        \item If $\det(D)>0$ and $\frac{\partial^2 f}{\partial \omega^2}<0$, then $(\omega_0,\varphi_0)$ is a local maximum point.
        \item If $\det(D)>0$ and $\frac{\partial^2 f}{\partial \omega^2}>0$, then $(\omega_0,\varphi_0)$ is a local minimum point.
        \item If $\det(D)<0$, then $(\omega_0,\varphi_0)$ is not an extremum point.
        \item If $\det(D)=0$, then $(\omega_0,\varphi_0)$ may be a local maximum or minimum point.
    \end{itemize}
    
    \item By comparing the function values at all possible local extremum points and boundary points, determine the global maximum $E(U,V)=\max (f(\omega,\varphi))$.
\end{enumerate}
\end{minipage}}
\end{center}

We provide a lemma concerning the property of Algorithm \ref{algorithm1} as follows:
\begin{lemma}
    \label{Algorithm properties-lemma}
    Given a unitary operator $U$ and a positive integer $N_0\geq 3$, if the minimum resources can be found within the search space $\mathcal{O}(N_0^3)$, then the execution process of Algorithm \ref{algorithm1} can be regarded as a function $h(\epsilon)$ related to the accuracy $\epsilon>0$. $h(\epsilon)$ is a monotonically decreasing function, and $h(\epsilon)\geq 3$ always holds. In particular, $h(\epsilon)=0$ if and only if $\epsilon=0$, and $U$ is a matrix representation of a combination of the Clifford and T gates, or the identity matrix.
\end{lemma}
\begin{proof}
    Let $\epsilon_2>\epsilon_1>0$, if $E\leq \epsilon_1$, then the minimum resources satisfying the condition is $h(\epsilon_1)$, which must satisfy the condition $E\leq \epsilon_2$. Thus, $h(\epsilon_1)\geq h(\epsilon_2)$. Conversely, if $E\leq \epsilon_2$, then the minimum resources satisfying the condition is $h(\epsilon_2)$, which does not necessarily satisfy the condition $E\leq \epsilon_1$. We need to continue to perform Algorithm \ref{algorithm1}, which ensures that $h(\epsilon_1)>h(\epsilon_2)$. Therefore, $h(\epsilon_1)\geq h(\epsilon_2)$, we conclude that $h(\epsilon)$ is a monotonically decreasing function.
    
    According to Theorem \ref{HTHT-theorem}, $n1,n2,n3$ are all positive integers, $h(\epsilon)\geq 3$ always holds. When $h(\epsilon)=0$, it means that $U$ does not need to be approximated, then $\epsilon=0$. Thus, $U$ is a matrix representation of a combination of the Clifford gate and the T gate, or the identity matrix.
\end{proof}

Based on Lemma \ref{Algorithm properties-lemma}, we propose a property theorem about approximating the decomposition of unitary operators around the $\hat{z}$ axis to analyze the minimum resources, as follows:
\begin{theorem}\label{Rz properties-theorem}
    $R_Z(\beta_0)$ is a known rotation operator, and the sum of the resources for approximating its decomposition into $n+1$ rotation operators $R_Z(\beta_0-\beta_1-\cdots-\beta_n)R_Z(\beta_1)\cdots R_Z(\beta_n)$ is minimum if and only if $\beta_1=\cdots=\beta_n=0$.
\end{theorem}
\begin{proof}
    For any given set of $\beta_1,\cdots,\beta_n$ and accuracy $\epsilon$, we consider the following cases:
    
        \textbf{Case 1:} $R_Z(\beta_0)=R_Z(\beta_0-\beta_1-\cdots-\beta_n)R_Z(\beta_1)\cdots R_Z(\beta_n)$, let $E(R_Z(\beta_0),V)<\epsilon,E(R_Z(\beta_0-\beta_1-\cdots-\beta_n),V_0)<\epsilon_0,E(R_Z(\beta_1),V_1)<\epsilon_1,\cdots,E(R_Z(\beta_n),V_n)<\epsilon_n,\epsilon_0+\epsilon_1+\cdots+\epsilon_n=\epsilon$, corresponding to the functions $h(\epsilon),h_0(\epsilon_0),h_1(\epsilon_1),\cdots h_n(\epsilon_n)$, respectively. Here, $V,V_0,\cdots,V_n$ are the actual approximation operators of $R_Z(\beta_0), R_Z(\beta_0-\beta_1-\cdots-\beta_n),R_Z(\beta_1),\cdots,R_Z(\beta_n)$ respectively.
        
        \textbf{Case 2:} $R_Z(\beta_0-\beta_1-\cdots-\beta_n)=R_Z(\beta_0)R_Z(-\beta_1)\cdots R_Z(-\beta_n)$, let $E(R_Z(\beta_0-\beta_1-\cdots-\beta_n),V_0')<\epsilon,E(R_Z(\beta_0),V')<\epsilon',E(R_Z(-\beta_1),V_1')<\epsilon_1',\cdots, E(R_Z(-\beta_n),V_n')<\epsilon_n',\epsilon'+\epsilon_1'+\cdots+\epsilon_n'=\epsilon$, corresponding to the functions $h_0(\epsilon),h(\epsilon'),h_1'(\epsilon_1'),\cdots,h_n'(\epsilon_n')$, respectively. Here, $V',V_0',\cdots,V_n'$ are the actual approximation operators of $R_Z(\beta_0), R_Z(\beta_0-\beta_1-\cdots-\beta_n),R_Z(-\beta_1),\cdots,R_Z(-\beta_n)$ respectively.
    
    Case 1 and Case 2 can be transformed into each other and are equivalent, meaning that a known rotation operator around the $\hat{z}$ axis can be decomposed into a set of rotation operators around the $\hat{z}$ axis. Thus, $h(\epsilon)\leq h_0(\epsilon_0)+h_1(\epsilon_1)+\cdots +h_n(\epsilon_n)$ for Case 1, and $h_0(\epsilon) \leq h(\epsilon')+h_1'(\epsilon_1')+\cdots +h_n'(\epsilon_n')$ for Case 2; or $h(\epsilon)\geq h_0(\epsilon_0)+h_1(\epsilon_1)+\cdots +h_n(\epsilon_n)$ for Case 1, and $h_0(\epsilon) \geq h(\epsilon')+h_1'(\epsilon_1')+\cdots +h_n'(\epsilon_n')$ for Case 2. Here, $h, h_0,\cdots,h_n, h_1' \cdots,h_n'$ depend only on rotation operators. $V, V_0,\cdots, V_n, V', V_0' \cdots,V_n'$ depend on rotation operators and accuracy. Additionally, $h(\epsilon)$ is a constant for different sets of $\beta_1,\cdots,\beta_n$.
    
    Now using the proof by contradiction, suppose $h_0(\epsilon) \geq h(\epsilon')+h_1'(\epsilon_1')+\cdots +h_n'(\epsilon_n')$, then $h_0(\epsilon)\geq h(\epsilon)+h_1'(\epsilon)+\cdots +h_n'(\epsilon)$ according to Lemma \ref{Algorithm properties-lemma}. Consequently, we have $h_0(\epsilon_0)+h_1(\epsilon_1)+\cdots +h_n(\epsilon_n) \geq h_0(\epsilon)+h_1(\epsilon)+\cdots +h_n(\epsilon) \geq h(\epsilon)+h_1'(\epsilon)+\cdots +h_n'(\epsilon)+h_1(\epsilon)+\cdots +h_n(\epsilon) \geq h(\epsilon)$, contradiction! In conclusion, $h_0(\epsilon) \leq h(\epsilon')+h_1'(\epsilon_1')+\cdots +h_n'(\epsilon_n')$ and $h(\epsilon)\leq h_0(\epsilon_0)+h_1(\epsilon_1)+\cdots +h_n(\epsilon_n)$. That is, the resources for directly approximating $R_Z(\beta_0)$ are minimum, achieving $h(\epsilon)$, if and only if $\beta_1=\cdots=\beta_n=0$. 
\end{proof}

\subsection{Complexity analysis}
We now focus on the exact lower bound problem of the T-count for the QFT. We propose the following theorem, which states that a controlled single-qubit unitary operator satisfying certain conditions can be decomposed into one CNOT gate and two single-qubit gates that are conjugate transposes of each other.
\begin{theorem}\label{Controlled-$(UX)$ decomposition-theorem}
    A unitary operator controlled-$(UX)$ can be decomposed into one CNOT gate and two single-qubit gates that are conjugate transposes of each other, if and only if $U=e^{i\alpha}R_z(\beta)R_y(\gamma)R_z(\beta)$.
\end{theorem}
\begin{proof}
    \textbf{Necessary condition:} $UX$ can written as $e^{i\alpha'/2}A'X(A')^{\dagger}$, then $U=e^{i\alpha'/2}A'X(A')^{\dagger}X$. According to Theorem \ref{Controlled-$U$ decomposition-theorem} and Eq. (\ref{Unitary operator decomposition-lemma-eq1}), up to a global phase, $A'\equiv R_z(\beta)R_y(\gamma/2)$, $A'^{\dagger}\equiv R_y(-\gamma/2)R_z(-\beta)$. Therefore, $U=e^{i\alpha}R_z(\beta)R_y(\gamma)R_z(\beta)$.
     
    \textbf{Sufficient condition:}
    $U=e^{i\alpha}R_z(\beta)R_y(\gamma)R_z(\beta)$, according to Eq. (\ref{Unitary operator decomposition-lemma-eq1}), then $\delta=\beta$. According to Theorem \ref{Controlled-$U$ decomposition-theorem}, let $A=R_z(\beta)R_y(\gamma/2), B=R_y(-\gamma/2)R_z(-\beta), C=I$, then $UX=e^{i\alpha}AXB$, where $B=A^{\dagger}$. Therefore, controlled-$(UX)$ can be decomposed into one CNOT gate and two single-qubit gates that are mutually conjugate transpose.
\end{proof}

\begin{proposition}\label{proposition4}
The method for decomposing the ancilla-free controlled-$R_k$, as shown in Figure \ref{fig:citco}, requires the minimum T-count when $A=I$ or $C=I$. The minimum T-count are the sum of those for approximating $R_z(-\pi /2^k)$ and $R_z(\pi /2^k)$.
\end{proposition}
\begin{proof}
    Suppose that there exists a decomposition method, as shown in Figure \ref{fig:adocu}. Here, $U=R_k$, $r$ is the number of CNOT gates, and $P$ is the phase shift gate.
    \begin{figure}[H]
    \centerline{\includegraphics[width=0.75\textwidth]{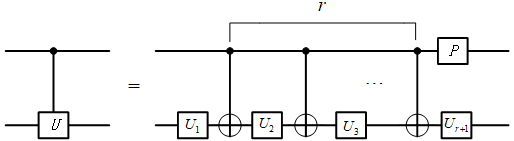}}
    \caption{General decomposition of an ancilla-free controlled-$U$ gate}
    \label{fig:adocu}
    \end{figure}
    When $r$ is an odd number, let $U_3=\cdots=U_{r+1}=I$, then $R_k=U_1XU_2X$. Therefore, $R_k=U'X$ can be decomposed into one CNOT and two single-qubit gates $U_1,U_2$ that are mutually conjugate transpose, where $U'=U_1XU_2$. Combining Eq. (\ref{Unitary operator decomposition-lemma-eq1}) and Eq. (\ref{controlled-Rk-eq}), we can obtain
    \begin{equation}\label{proposition3-eq}
        R_k=\begin{bmatrix} 
             1&0 \\ 
             0&e^{2\pi i/2^k} \end{bmatrix}=\begin{bmatrix}
        e^{i(\alpha-\beta/2-\delta/2)}\cos{\frac{\gamma}{2}} & -e^{i(\alpha-\beta/2+\delta/2)}\sin{\frac{\gamma}{2}}\\
        e^{i(\alpha+\beta/2-\delta/2)}\sin{\frac{\gamma}{2}} & e^{i(\alpha+\beta/2+\delta/2)}\cos{\frac{\gamma}{2}}
        \end{bmatrix}.
    \end{equation}
    Thus, we have $\alpha=\pi /2^k, \beta+\delta=\pi/2^{k-1}, \gamma=0$, then $R_k=e^{i\pi/2^k}R_z(\pi/2^{k-1})$. According to Theorem \ref{Controlled-$(UX)$ decomposition-theorem}, $R_kX=e^{i\pi/2^k}R_z(\pi/2^{k-1})R_y(\pi)R_z(-\pi)$ is satisfied if and only if $k=1$, i.e., $R_k=Z$. 
    
    Therefore, $r$ can only be an even number for $k=2,\cdots,n$. $U_1,\cdots,U_{r+1}$ are the operators $R_z(\theta_1),\cdots,$ $R_z(\theta_{r+1})$ that rotate around the $\hat{z}$ axis, then
    \begin{equation}
    \left\{ 
    \begin{aligned}
        &U_1U_2\cdots U_{r+1}=I \Longleftrightarrow \theta_1+\theta_2+\cdots+\theta_r+\theta_{r+1}=0,\\
        &U_1XU_2X\cdots XU_rXU_{r+1}=R_z(\beta+\delta) \Longleftrightarrow \theta_1-\theta_2+\cdots-\theta_r+\theta_{r+1}=\beta+\delta,
    \end{aligned}
    \right.
    \end{equation}
thus, $\theta_1+\theta_3+\cdots+\theta_{r+1}=\frac{\beta+\delta}{2}$, $\theta_2+\theta_4+\cdots+\theta_r=-\frac{\beta+\delta}{2}$. According to Theorem \ref{Rz properties-theorem}, if and only if there exist a certain $\theta_{i}=\frac{\beta+\delta}{2}$ for $i=2\ell-1$, where $\ell=1,\cdots,\frac{r}{2}+1$, a certain $\theta_{j}=-\frac{\beta+\delta}{2}$ for $j=2\ell$, where $\ell=1,\cdots,\frac{r}{2}$, the resources reach their minimum. At this point, $r=2$, this decomposition method corresponds to the one shown in Figure \ref{fig:citco}, where $A=I$ or $C=I$. From Eq. (\ref{proposition3-eq}), $A_k=R_z(\beta), B_k=R_z(-\pi /2^k), C_k=R_z(\pi /2^k-\beta), \beta\in[0,2\pi)$ for each $k=2,\cdots,n$. Let $\beta=0, \text{ or } \pi/2^k$, $A_K=I,C_k=R_z(\pi /2^k)$, or $C_K=I,A_k=R_z(\pi /2^k)$, then the resources reach their minimum, i.e., the sum of those for approximating $R_z(-\pi /2^k)$ and $R_z(\pi /2^k)$. 
\end{proof}

When $k$ is fixed (the controlled-$R_k$ is determined), then $B_k$ is determined, while $A_k,C_k$ are not unique and are related to $\beta$. To obtain the minimum resources, it is theoretically necessary to traverse $\beta\in (0,2\pi)$ and then compute the sum of the minimum resources $\min n_{(A_k+C_k)}$ for approximating $A_k,C_k$.
However, it is interesting that according to Proposition \ref{proposition4}, the minimum sum of the resources for approximating $A_k,C_k$ and $B_k$ is equivalent to using Algorithm \ref{algorithm1} to compute the sum of those for approximating $R_z(\pi /2^k)$ and $R_z(-\pi /2^k)$. According to Remark \ref{rm1}, it follows that the minimum T-count is twice their resources. The exact lower bound problem of the T-count for the QFT can be considered as the problem of
determining an integer $N_0$ such that the minimum T-count for approximating $R_z(\pi /2^k)$ or $R_z(-\pi/2^k)$ with given accuracy $\epsilon$ can be found in the search space $\mathcal{O}(N_0^3)$, for all $k=3,\cdots,n$. We call the problem $\text{ELBP}_{\text{T-count}}$. 

To show the NP-completeness of $\text{ELBP}_{\text{T-count}}$, we recall the K-SAT problem (($K\geq 3$)) \cite{karp2010reducibility}. Given a finite set of boolean variables $X=\{x_1,x_2,\cdots,x_n\},|X|=n$, a set of clauses $C=\{C_1,C_2,\cdots,C_m\},|C|=m,C=C_1 \land C_2 \land \cdots \land C_m$, each $C_i$ is a disjunctive normal form consisting of $K$ variables. Consider whether there exists a truth value assignment for a set of Boolean variables such that $C$ is true. It is known that the K-SAT problem is a class of NP-complete problems \cite{cook2023complexity,karp2010reducibility}, hence we can use a reduction from K-SAT problem to prove NP-completeness of $\text{ELBP}_{\text{T-count}}$. We propose the following theorem:

\begin{theorem}\label{Reduction-theorem}
     $\text{ELBP}_{\text{T-count}}$ is NP-complete. 
\end{theorem}
\begin{proof}
    Given an integer $N_0\geq 3$, the search space is $\sum\limits_{n_{sum}=3}^{N_0}\sum\limits_{n_1=1}^{n_{sum}-2}(n_{sum}-n_1-1)=\frac{1}{6}N_0(N_0-1)(N_0-2)=\mathcal{O}(N_0^3)$ according to steps 4-6 in Algorithm \ref{algorithm1}. As shown in Theorem \ref{sk-theorem}, $N_0=\mathcal{O}(\log^{c}(1/\epsilon))$, so the search space is logarithmic polynomial.
    
    Given an accuracy $\epsilon$ and a positive integer $N_1=\mathcal{O}(\log^{c}(1/\epsilon))$, let $K=\frac{1}{6}N_1(N_1-1)(N_1-2)=\mathcal{O}(N_1^3)$. An instance of the K-SAT problem is $(\ell_{31}\lor \ell_{32} \lor \cdots \lor \ell_{3K})\land (\neg\ell_{31}\lor \neg\ell_{32} \lor \cdots \lor \neg\ell_{3K})\land \cdots \land (\ell_{n1}\lor \ell_{n2} \lor \cdots \lor \ell_{nK})\land (\neg\ell_{n1}\lor \neg\ell_{n2} \lor \cdots \lor \neg\ell_{nK})$, we construct an instance of $\text{ELBP}_{\text{T-count}}$, where a finite set of boolean variables $L=\{\ell_{31},\ell_{32},\cdots,\ell_{3K},\cdots, \ell_{n1},\ell_{n2},\cdots,\ell_{nK}\}$, $|L|=(n-2)K$, a set of clauses $C=\{C_{31},C_{32},C_{41},C_{42},\cdots,C_{n1},C_{n2}\}$, $|C|=2(n-2)$, $C=C_{31} \land C_{32}\land C_{41}\land C_{42} \land \cdots \land C_{n1}\land C_{n2}$, each $C_{k1}$ or $C_{k2}$, $k=3,\cdots,n$ is a disjunctive normal form consisting of $K$ variables. Here, $\lor$ means taking the minimum value "$\min$", and $\land$ means taking the maximum value "$\max$". 
    
    If the variable $\ell_{kj}$ or $\neg\ell_{kj}$ is true for $k=3,\cdots,n,j=1,\cdots,K$, then the error $E_{kj}$ is recorded, otherwise it is assigned $\infty$ and recorded, where $E_{kj}$ for a certain $k$ is the error for a certain $R_z(\pi /2^k)$ or $R_z(-\pi /2^k)$  and differnet $j$ corresponds to different $V$ in the search space $\mathcal{O}(N_1^3)$ using Algorithm \ref{algorithm1}. Let $\epsilon'=\max (E_{31},E_{32},\cdots,E_{3K},\cdots,E_{n1},E_{n2},\cdots,E_{nK})$. If each clause is true, then at least one variable in each clause is true and $\epsilon' \neq \infty$; Otherwise, at least one clause is false, and then $\epsilon'=\infty$. Therefore, the boolean formula is satisfiable if and only if $\epsilon' \neq \infty$. This means that the minimum T-count for approximating $R_z(\pi /2^k)$ or $R_z(-\pi/2^k)$ with the accuracy $\epsilon'$ can be found in the search space $\mathcal{O}(N_1^3)$ using Algorithm \ref{algorithm1}, for all $k=3,\cdots,n$. Determine $N_0=N_1$ when $\epsilon' < \epsilon$, then 
    the minimum T-count for approximating $R_z(\pi /2^k)$ or $R_z(-\pi/2^k)$ can be found in the search space $\mathcal{O}(N_0^3)$; otherwise, it cannot.
    
     Clearly $\text{ELBP}_{\text{T-count}}$ is in NP, as the mathematical computation in step 9 of Algorithm \ref{algorithm1} is polynomial-time computable when given an $N_0$, and hence step 10 can be efficiently verified. It suffices to show NP-hardness by reducing the K-SAT problem to $\text{ELBP}_{\text{T-count}}$, and hence $\text{ELBP}_{\text{T-count}}$ is NP-complete.
\end{proof}

\section{The exact lower bound of CNOT-count for the fault-tolerant QFT}
From section 3, we known that $\text{ELBP}_{\text{T-count}}$ is an NP-complete problem and the fault-tolerant CNOT-count is positively correlated with the T-count. Therefore, the exact lower bound problem of CNOT-count for the fault-tolerant QFT is also NP-complete. Nevertheless, we can still compute its exact lower bound under partial fault-tolerant accuracy. We now turn to the transversal implementation of universal quantum gates with the minimum CNOT-count.
\subsection{Universal fault-tolerant quantum gates with the minimum CNOT-count}
The universal gates can be implemented transversely without requiring a fault-tolerant construction in FTQC except for the T gate equivalent to the rotation operator $R_z(\pi/4)$, so we need to consider the fault-tolerant construction of Z-rotation gates. Similar to the fault-tolerant construction of the T gate in \cite{nielsen2010quantum}, we provide a general fault-tolerant construction for $U=e^{i\theta/2}R_z(\theta)$ that cannot be implemented transversely without requiring a fault-tolerant construction, as shown in Figure \ref{fig:coftsqqg}.

\begin{figure}[H]
\centerline{\includegraphics[width=0.53\textwidth]{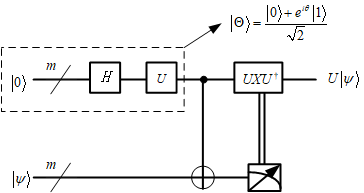}}
  \caption{Construction of fault-tolerant single-qubit quantum gate}
   \label{fig:coftsqqg}
\end{figure}
\noindent Here, the auxiliary state $|\Theta\rangle$ is a +1 eigenstate of the operator $UXU^{\dagger}=R_z(2\theta)X$.

Let $M=e^{i\theta}R_z(2\theta)$, we observe the single-qubit operator $MX$ with eigenvalue $\pm 1$, and its fault-tolerant measurement is shown in Figure \ref{fig:ftmoaom}.
\begin{figure}[H]
\centerline{\includegraphics[width=0.74\textwidth]{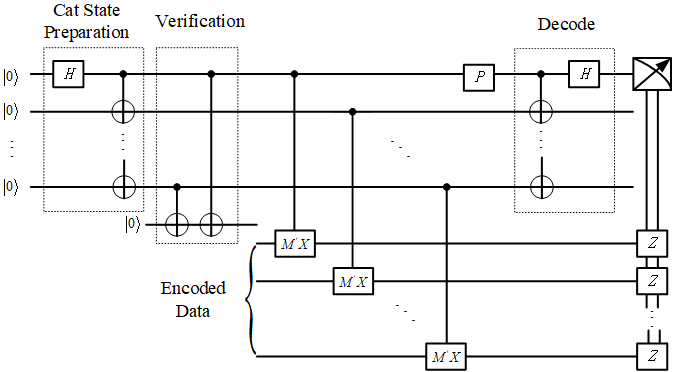}}
  \caption{Fault-tolerant measurement of an encoded observable $M$}
   \label{fig:ftmoaom}
\end{figure}
\noindent Here, $P$ is a phase shift gate, $P=e^{-i\theta/2}R_z(-\theta)$, and $M'$ is the fault-tolerant operator that can be implemented transversely for $M$ without requiring a fault-tolerant construction. The framed parts are the preparation, verification, and final decoding of the cat state $|Cat\rangle=\frac{1}{\sqrt{2}}(|0_C\rangle+|1_C\rangle)$. If the cat state is successfully prepared, then

\begin{equation}\label{Fault-tolerant measurement-eq}
    \begin{aligned}
   |Cat\rangle|0_L\rangle  
    &\xrightarrow{C(M'X)}\frac{1}{\sqrt{2}}(|0_C\rangle|0_L\rangle+e^{2i\theta}|1_C\rangle|1_L\rangle)\\
    &\xrightarrow{P}\frac{1}{\sqrt{2}}(|0_C\rangle|0_L\rangle+e^{i\theta}|1_C\rangle|1_L\rangle)\\
    &= \frac{1}{\sqrt{2}}\left( \frac{|0_C\rangle+|1_C\rangle}{\sqrt{2}}\frac{|0_L\rangle+e^{i\theta}|1_L\rangle}{\sqrt{2}}+ \frac{|0_C\rangle-|1_C\rangle}{\sqrt{2}}\frac{|0_L\rangle-e^{i\theta}|1_L\rangle}{\sqrt{2}} \right)\\
     &\xrightarrow{Decode}\frac{1}{\sqrt{2}}\left( |0\rangle\frac{|0_L\rangle+e^{i\theta}|1_L\rangle}{\sqrt{2}}+ |1\rangle\frac{|0_L\rangle-e^{i\theta}|1_L\rangle}{\sqrt{2}} \right),
    \end{aligned}
\end{equation}
where $|0_L\rangle$ and $|1_L\rangle$ denote the encoding states of logical $|0\rangle$ and $|1\rangle$, respectively.

We use the above fault-tolerant measurement method to prepare the auxiliary state $|\Theta\rangle$. If the measurement result is $+1$, it can be considered to have been prepared correctly; if it is $-1$, a fault-tolerant $Z$ operation needs to be applied to change the state.

Next, we propose a proposition regarding the minimum CNOT-count used in the general fault-tolerant construction in Figure \ref{fig:coftsqqg} and Figure \ref{fig:ftmoaom}. The proposition is as follows:
\begin{proposition}\label{proposition3}
The fault-tolerant construction of the single-qubit gate $U$ rotating around the $\hat{z}$ axis, as shown in Figure \ref{fig:coftsqqg}, requires the minimum CNOT-count. The preparation of the auxiliary state, as shown in Figure \ref{fig:ftmoaom}, also requires the minimum CNOT-count.
\end{proposition}
\begin{proof}
For $U=e^{i\theta/2}R_z(\theta)$ that cannot be implemented transversally without requiring a fault-tolerant construction, auxiliary qubits are required. We use only $m$ CNOT gates for transversal implementation to swap the auxiliary qubit $|0\rangle$ with the data qubit $|\psi\rangle$ when one logical qubit is encoded into $m$ physical qubits, as shown in Figure \ref{fig:slqaaq}.
    \begin{figure}[H]
\centerline{\includegraphics[width=0.55\textwidth]{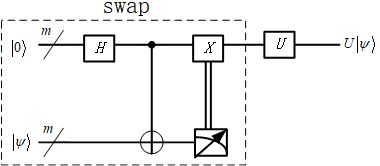}}
  \caption{Swap data qubits and auxiliary qubits}
   \label{fig:slqaaq}
\end{figure}

Here, it uses the minimum CNOT-count to implement the swap operation, since it is impossible to implement the interaction between two qubits using only single-qubit gates.

\begin{remark}\label{rm2}
The swap operation in Figure \ref{fig:slqaaq} only implements the exchange of data qubits to auxiliary qubits by measuring the original data qubits. In general, at least three CNOT gates are required when implementing the swap operation between two arbitrary quantum states, as shown in Figure \ref{fig:fuqg}.
\end{remark}

Applying the relations $UXU^{\dagger}=R_z(2\theta)X$ and $(U\otimes I) \cdot CNOT=CNOT \cdot (U\otimes I)$, we can obtain the fault-tolerant construction of $U$ in Figure \ref{fig:coftsqqg}.
Let $|\psi\rangle=a|0\rangle+b|1\rangle$, then perform a fault-tolerant CNOT operation, giving
\[\frac{1}{\sqrt{2}}\left[|0\rangle(a|0\rangle+b|1\rangle)+e^{i\theta}|1\rangle(a|1\rangle+b|0\rangle)\right] =\frac{1}{\sqrt{2}}\left[(a|0\rangle+be^{i\theta}|1\rangle)|0\rangle+(b|0\rangle+ae^{i\theta}|1\rangle)|1\rangle\right].\]
Finally, measure the second qubit. If the result is 0, the process is complete; otherwise, apply the $UXU^{\dagger}$ operation to the remaining qubits.

In Figure \ref{fig:ftmoaom}, the preparation of the auxiliary state must rely on the auxiliary qubits, and the preparation of the cat state has already reached the minimum CNOT-count using $m-1$ CNOT gates. According to Theorem \ref{Controlled-$(UX)$ decomposition-theorem}, the fault-tolerant controlled-$M'X$ can be decomposed into $m$ CNOT gates for transverse implementation and other single-qubit gates, which has reached the minimum CNOT-count. Therefore, we use the minimum CNOT-count to make $U$ transverse implementation.
\end{proof}

From the above analysis, controlled-$M$ must be implemented transversely without requiring a fault-tolerant construction in theory, where $M=e^{i\theta}R_z(2\theta)$. For $R_k=e^{i\pi/2^k}R_z(\pi/2^{k-1})$, it is clear that for $\theta=\frac{\pi}{2^{k-1}}$ with $k > 3$, this condition is not satisfied. In particular, when $\theta=\frac{\pi}{4}$ with $k = 3$, then $U=T,M=S,M'=ZS$. According to Theorem \ref{Controlled-$(UX)$ decomposition-theorem}, the controlled-$M'X$ can be decomposed into one CNOT gate and two single-qubit gates $A=R_z(\frac{3}{4}\pi)=e^{-3i\pi/8}TS$ and $B=R_z(-\frac{3}{4}\pi)=e^{3i\pi/8}S^{\dagger}T^{\dagger}$, where $T^{\dagger}=T^7, S^{\dagger}=S^3$. Fortunately, the standard set of universal gates \{H, S, T, CNOT\} can construct all physical quantum gates. 

Therefore, combining Figure \ref{fig:coftsqqg} and Figure \ref{fig:ftmoaom}, at least $4m$ fault-tolerant CNOT gates are required to transversely implement the fault-tolerant T gate (without considering the encoding circuits).

\subsection{Result analysis}
We set the accuracy $\epsilon$ to the current maximum fault-tolerant accuracy $10^{-2}$ \cite{fowler2009high, knill2005quantum} in FTQC and compute the minimum resources $\min n_{sum}$ for approximating $R_z(\pi /2^k)$ and $R_z(-\pi /2^k)$ for different values of $k$ using  Algorithm \ref{algorithm1}. According to Remark \ref{rm1}, the minimum T-count is $2\min n_{sum}$. The minimum T-count for approximating $R_z(\pi /2^k)$ and $R_z(-\pi /2^k)$ for different $k$ are shown in Table \ref{tab:1}.
\begin{table}[H]
\footnotesize
\centering
\caption{Minimum T-count with accuracy $\epsilon=10^{-2}$ for different $k(>2)$}
\label{tab:1}
\tabcolsep 6pt 
\begin{tabular}{ccccc}
\toprule
 & $R_z(-\pi /2^k)-\min n_{sum}$ & $R_z(\pi /2^k)-\min n_{sum}$ & $R_z(\pi /2^k),R_z(-\pi /2^k)-\min n_{sum}$ & Minimum T-count\\
\hline
$k=3$ & 57   & 218  & 275 & $550$\\
$k=4$ & 275  & 195  & 470 & $940$\\
$k=5$ & 149  & 172  & 321 & $642$\\
$k=6$ & 149  & 172  & 321 & $642$\\
$k=7$ & 321  & 172  & 493 & $986$\\
$k=8$ & 321  & 172  & 493 & $986$\\
$k=9$ & 493  & 493 & 986 & $1972$\\
$\vdots$ & $\vdots$ & $\vdots$ & $\vdots$ & $\vdots$\\
$k=4096$ & 493  & 493  & 986 & $1972$\\
\bottomrule
\end{tabular}
\end{table}

In Table \ref{tab:1}, when $k\geq 9$, the minimum T-count remains unchanged at 1972. This is because as $k$ increases, $R_z(\pi /2^k)$ and $R_z(-\pi /2^k)$ gradually approach the identity matrix and become insufficiently sensitive to this accuracy. When $k=2$, the controlled-$R_2$ is the controlled-$S$. According to Theorem \ref{Controlled-$U$ decomposition-theorem}, it can be decomposed into one T gate, one T$^{\dagger}$ gate, and two CNOT gates. The $T^{\dagger}$ can be approximated using Algorithm \ref{algorithm1} and the minimum resources $\min n_{sum}$ are 206 with $\epsilon=10^{-2}$, i.e., the minimum T-count is 412. Compared with $T^{\dagger}=T^7$, only seven T gates are required with zero error. Therefore, when $k=2$, the minimum T-count for controlled-$S$ is 8.

We analyze the minimum CNOT-count for the fault-tolerant QFT with different lengths. If the input length of the QFT is $n$, there are $n(n-1)/2$ controlled-$R_k$ in Figure \ref{fig:eqcftqft}, including $(n-1)$ controlled-$R_2$, $(n-2)$ controlled-$R_3$, $\cdots$, two controlled-$R_{n-1}$, and one controlled-$R_n$. According to Proposition \ref{proposition4}, the decomposition method for approximating the controlled-$R_k$ with the minimum T-count is to decompose it into one $R_z(\pi/2^k)$, one $R_z(-\pi/2^k)$ and two CNOT gates, where the minimum T-count for approximating $R_z(\pi /2^k)$ and $R_z(-\pi /2^k)$ are shown in Table \ref{tab:1}. At this point, the required fault-tolerant CNOT-count is also minimum. Additionally, the final stage of the QFT requires $\lfloor \frac{n}{2} \rfloor$ swap gates, which can be implemented by $3\cdot \lfloor \frac{n}{2} \rfloor$ CNOT gates without any auxiliary qubits, according to Remark \ref{rm2}. 

From section 4.1, we know that at least $4m$ fault-tolerant CNOT gates are required for the transversal fault-tolerant T gate, while the logical CNOT gate can be implemented transversally without requiring a fault-tolerant construction, using only $m$ fault-tolerant CNOT gates. Therefore, the exact lower bound  of CNOT gates ($\text{ELB}_{\text{CNOT-count}}$) for the fault-tolerant QFT with different input lengths $n$ is

\begin{equation}
   num(CNOT)=\left(\frac{n(n-1)}{2}\cdot 2+\lfloor\frac{n}{2}\rfloor\cdot 3\right)m+num(T)\cdot 4m , 
\end{equation}
where $num(T)$ is the sum of the minimum T-count with $k=2,3,\cdots,n$. When the accuracy is at most $10^{-2}$, the $\text{ELB}_{\text{CNOT-count}}$ for the fault-tolerant QFT with different $n$ is shown in Figure \ref{fig:QFT-FTCNOT-LBfdwa}.
\begin{figure}[H]
    \centerline{\includegraphics[width=0.95\textwidth]{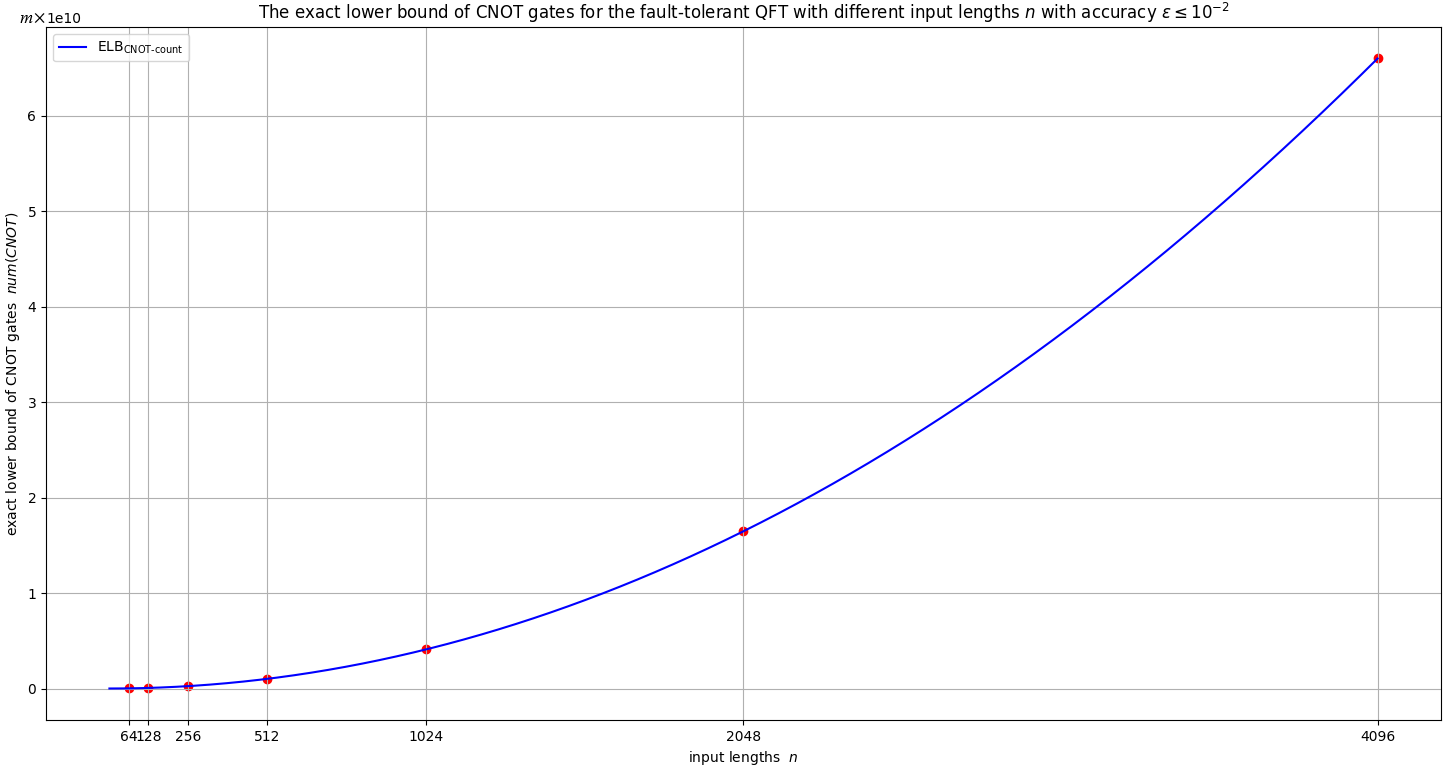}}
    \caption{The $\text{ELB}_{\text{CNOT-count}}$ for the fault-tolerant QFT with different input lengths $n$ at accuracy $\epsilon \leq 10^{-2}$}
    \label{fig:QFT-FTCNOT-LBfdwa}
    \end{figure}

The operation time of fault-tolerant CNOT gates is limited by inherent physical limitations. Especially, in ion trap quantum computers, the CNOT operation between two qubits utilizes collective excitation particles such as phonons to transmit interactions, which makes the operation efficiency limited by the propagation speed of media such as phonons. Moreover, the CNOT gates in ion traps can only be operated serially. Even if different CNOT gates involve different qubits, they cannot be operated in parallel. Therefore, CNOT gates significantly affect the operation time of quantum algorithms. In \cite{yang2013post}, Yang et al. for the first time estimated the average lower bound of the time for a single physical CNOT operation in an ion trap quantum computer by analyzing the phonon speed, which is $2.85 \times 10^{-4}$ s. 

In cryptographic systems, common input lengths include 64-bit, 128-bit, 256-bit, 512-bit, 1024-bit, 2048-bit, and even 4096-bit. We particularly compute their $\text{ELB}_{\text{CNOT-count}}$ and estimate the lower bound of their time for the fault-tolerant QFT, as shown in Table \ref{tab:2}. 

\begin{table}[H]
\footnotesize
\centering
\caption{The $\text{ELB}_{\text{CNOT-count}}$ for the fault-tolerant QFT with common input lengths $n$ at accuracy $\epsilon \leq10^{-2}$}
\label{tab:2}
\tabcolsep 18pt 
\begin{tabular}{ccccc}
\toprule
 & T gates & Fault-tolerant CNOTs & Times (In particular, $m=7$) \\
\hline
$n=64$ &  3429044  &  13720304$m$  & $2.737\times 10^{4}$s/0.317 days\\
$n=128$ & 14902708 & 59627280$m$  & $1.190\times 10^{5}$s/1.377 days\\
$n=256$ & 62081972  & 248393552$m$ & $4.955\times 10^{5}$s/5.735 days\\
$n=512$ & 253368244  & 1013735376$m$ & $2.022\times 10^{6}$s/23.407 days\\
$n=1024$ & 1023651764  & 4095656144$m$  & $8.171\times 10^{6}$s/94.570 days\\
$n=2048$ & 4115062708  & 16464446160$m$  & $3.285\times 10^{7}$s/380.169 days\\
$n=4096$ & 16501260212 & 66021820112$m$  & $1.317\times 10^{8}$s/1524.462 days\\
\bottomrule
\end{tabular}
\end{table}

In particular, when $m=7$, taking the famous Steane code as an example \cite{steane1996error}. For shorter input lengths $n$, such as some lightweight cryptography, the time required to operate QFT on a quantum computer is very short, while for $n=2048$ and $n=4096$, the time is relatively long, 380.169 days and 1524.462 days respectively, which is about more than one year and more than four years. This analysis of $\text{ELB}_{\text{CNOT-count}}$ can provide a reference for the idea of active defense in a quantum setting.

The search space of the minimum resources for approximating $R_z(\pi /2^k)$ or $R_z(-\pi /2^k)$ is $\mathcal{O}(N_0^3)$ for all $k=3,\cdots,n$. In fact, as the accuracy decreases, it is challenging to determine an integer $N_0$ such that the minimum T-count for approximating them with given accuracy $\epsilon$ can be found within the search pace $\mathcal{O}(N_0^3)$ in polynomial time, given an integer $N_0$. For example, given the integer $N_0=2^{11}$, we have verified that the minimum T-count cannot be found in the search space $\mathcal{O}(2^{33})$ when the accuracy is $\epsilon=10^{-3}$, which implies that the search space of the minimum T-count may be much larger than $\mathcal{O}(2^{33})$. This entire process is actually very time-consuming and requires continuously attempting to determine the value of $N_0$ so that the minimum T-count is found for all Z-rotation operators.

\section{Discussion}
As shown in Theorem \ref{sk-theorem}, the Solovay–Kitaev theorem \cite{kitaev2002classical} states that any single-qubit gate can be approximated to an accuracy $\epsilon$ with $\mathcal{O}(\log ^c(1/\epsilon))$ gates from the discrete universal set. This number of gates grows polylogarithmically with decreasing accuracy, which is probably acceptable for almost all practical applications. It has been proven in \cite{nielsen2010quantum} that the value of $c$ cannot be less than 1 and lies between 1 and 2, close to 2, though determining the best possible value remains an open problem. We believe that the problem is at least an NP-hard problem. When the accuracy $\epsilon$ is small enough, there exists an accuracy $\epsilon_0$ and a constant $C'$, such that when $\epsilon<\epsilon_0$, Algorithm \ref{algorithm1} can be regarded as a function $h(\epsilon)\leq C'\log ^c(1/\epsilon)$, and the smaller $\epsilon$ is, the closer $h(\epsilon)$ is to $C'\log ^c(1/\epsilon)$. Therefore, if $\epsilon_2<\epsilon_1\ll \epsilon_0$, $\frac{h(\epsilon_2)}{h(\epsilon_1)}\approx \left(\frac{\log (1/\epsilon_2)}{\log (1/\epsilon_1)} \right)^c$, where $h(\epsilon_1),h(\epsilon_2)$ can be computed based on Algorithm \ref{algorithm1}. The K-SAT problem, including 
four clauses, can be reduced to this problem of determining an positive integer $N_0$ such that $h(\epsilon_1)$ and $h(\epsilon_2)$ exist for approximating $U$ with accuracy $\epsilon_1$ and $\epsilon_2$ respectively in the search space $\mathcal{O}(N_0^3)$, given a single-qubit unitary operator $U$, thereby this problem is at least as hard as the K-SAT problem. As a subroutine for computing the value of $c$, we believe that determining the best possible value is at least NP-hard.

The unique structure of the QFT circuit makes finding its optimal circuit relatively easy. However, equivalence check is QMA-complete \cite{janzing2005non,bookatz2012qma}, meaning that verifying whether two different circuits yield the same unitary transformation is QMA-complete. Quantum circuit optimization involves finding its optimal circuit in the set of different circuits with the equivalence class of a certain unitary transformation. Therefore, the circuit optimization problems, including but not limited to quantum algorithms based on the QFT,  are at least QMA-hard, i.e., given a quantum circuit, it is challenging to find its optimal version. Specifically, optimizing the circuit for Shor's algorithm is at least QMA-hard. Focusing on optimizing the CNOT-count, optimizing quantum circuits for arithmetic operations such as modular addition, modular multiplication, and modular exponentiation that are fundamental to quantum circuit decomposition, is a challenging task that typically requires continuous optimization during the design stage.

The order-finding and factoring problems based on the QFT provide evidence that quantum computers may be more powerful than classical computers, posing a credible challenge to the strong Church–Turing thesis. Our work can introduce new connotations for quantum adversaries. From a practical point of view, if efficient algorithms based on the QFT can be implemented within a meaningful time frame on a quantum computer, then they can be used to break some cryptosystems such as RSA. Otherwise, the development of quantum computers is unlikely to threaten the security of those classical cryptosystems.

\section{Conclusion}
In this paper, we study the exact lower bound problem of CNOT gate complexity for fault-tolerant QFT. We first analyze the complexity of $\text{ELBP}_{\text{T-count}}$ and have shown that this problem is NP-complete. When the ancilla-free controlled-$R_k$ is decomposed in different ways, it generates different single-qubit gates in addition to CNOT gates. As a generalization, we propose an algorithm to exactly compute the minimum T-count for approximating any single-qubit gate with any given accuracy. This algorithm combines numerical and analytical methods to avoid traversing all quantum states in the normalized space and exactly determine whether the error conditions are satisfied. We then provide a property of the proposed algorithm to analyze the minimum resources, aiming to show the NP-completeness of $\text{ELBP}_{\text{T-count}}$ and compute the $\text{ELB}_{\text{CNOT-count}}$ for the fault-tolerant QFT. 

Due to the NP-completeness of $\text{ELBP}_{\text{T-count}}$, it would appear that computing the $\text{ELB}_{\text{CNOT-count}}$ for the fault-tolerant QFT with any given accuracy is intractable. Nevertheless, we can still compute the $\text{ELB}_{\text{CNOT-count}}$ under partial fault-tolerant accuracy. We have proved that the transversal implementation of universal quantum gates reaches the minimum CNOT-count. Furthermore, we approximate the Z-rotation gates after decomposing the ancilla-free controlled-$R_k$ with the current maximum fault-tolerant accuracy $10^{-2}$ and provide the $\text{ELB}_{\text{CNOT-count}}$ with different input lengths. In particular, we estimate the lower bound of the effective execution time for the QFT based on Steane code on ion trap computers. For shorter input lengths $n$, such as some lightweight cryptography, the time required to operate the QFT on an ion trap computer is very short. When $n=2048$ and $n=4096$, the time required is relatively long compared to lightweight cryptography, which is 380.169 days and 1524.462 days respectively.

Finally, we discuss that determining the best possible value of $c$ in $\mathcal{O}(\log^c(1/\epsilon))$ implied by the Solovay–Kitaev theorem is at least an NP-hard problem, and the circuit optimization problem, such as Shor's algorithm, is at least a QMA-hard problem. Our work can introduce new connotations for quantum adversaries and provide a reference for the idea of active defense based on the QFT.

\section*{Acknowledgement}
This work was supported by the Beijing Natural Science Foundation (Grant No.4234084).

\bibliographystyle{unsrt}
\bibliography{bib}
 
\end{document}